\documentclass[11pt,reqno]{amsart}
\usepackage{amsmath,amssymb,latexsym,mathrsfs}
\usepackage{graphicx}
\usepackage{fancyhdr,amssymb}
\usepackage{url}		
\usepackage{fixltx2e,amsmath,graphicx,amssymb}
\usepackage[mathlines]{lineno}
\usepackage{float}
\usepackage{oubraces}
\usepackage{tabularx}
\usepackage{multirow}
\usepackage{tikz}
\usepackage{here}
\usepackage{cite}

\newcommand{\N}{{\mathbb N}}

\numberwithin{equation}{section}

\theoremstyle{plain}
\newtheorem{theorem}{Theorem}
\newtheorem{corollary}[theorem]{Corollary}
\newtheorem{lemma}[theorem]{Lemma}
\newtheorem{proposition}[theorem]{Proposition}

\theoremstyle{definition}
\newtheorem{definition}[theorem]{Definition}
\newtheorem{example}[theorem]{Example}
\newtheorem{conjecture}[theorem]{Conjecture}

\theoremstyle{remark}

\renewcommand{\epsilon}{\varepsilon}
\renewcommand{\phi}{\varphi}

\DeclareMathOperator{\Suff}{Suff}
\DeclareMathOperator{\Pref}{Pref}
\DeclareMathOperator{\Fact}{Fact}

\DeclareMathOperator{\AS}{AS}
\newcommand{\ass}[2]{\AS_{#2}(#1)}
\newcommand{\as}[1]{\AS(#1)}

\def\cd3#1{\textbf{\textsf{#1}}}
\def\sa#1{\cd3{#1}}

\begin{document}
\sloppy

\setcounter{page}{1}

\title[Abelian-Square-Rich Words]{Abelian-Square-Rich Words}

\author[G. Fici]{Gabriele Fici}
\address[G. Fici]{Dipartimento di Matematica e Informatica\\
                Universit\`a di Palermo\\
                Palermo, Italy}
\email[Corresponding author]{gabriele.fici@unipa.it}
\thanks{Some of the results contained in this paper were presented 
(without the third author) at the 10th International Conference on Words, WORDS 2015 \cite{FiMi15}.}

\author[F. Mignosi]{Filippo Mignosi}
\address[F. Mignosi]{Dipartimento di Ingegneria e Scienze dell'Informazione e Matematica\\ Universit\`a dell'Aquila\\ 
   L'Aquila, Italy}
\email{filippo.mignosi@univaq.it}

\author[J. Shallit]{Jeffrey Shallit}
\address[J. Shallit]{School of Computer Science, University of Waterloo\\ 
   Waterloo, ON N2L 3G1, Canada}
\email{shallit@cs.uwaterloo.ca}
   
\begin{abstract}
An abelian square is the concatenation of two words that are anagrams of one another. A word of length $n$ can contain at most $\Theta(n^2)$ distinct factors, and there exist words of length $n$ containing $\Theta(n^2)$ distinct abelian-square factors, that is, distinct factors that are abelian squares. This motivates us to study infinite words such that the number of distinct abelian-square factors of length $n$ grows quadratically with $n$. More precisely, we say that an infinite word $w$ is {\it abelian-square-rich} if, for every $n$, every factor of $w$ of length $n$ contains, on average, a number of distinct abelian-square factors that is quadratic in $n$; and {\it uniformly abelian-square-rich} if every factor of $w$ contains a number of distinct abelian-square factors that is proportional to the square of its length. Of course, if a word is uniformly abelian-square-rich, then it is abelian-square-rich, but we show that the converse is not true in general.
We prove that the Thue-Morse word is uniformly abelian-square-rich and that the function counting the number of distinct abelian-square factors of length $2n$ of the Thue-Morse word is $2$-regular. As for Sturmian words, we prove that a Sturmian word $s_{\alpha}$ of angle $\alpha$ is uniformly abelian-square-rich if and only if the irrational $\alpha$ has bounded partial quotients, that is, if and only if $s_{\alpha}$ has bounded exponent. 
\end{abstract}

\maketitle

\section{Introduction}

A fundamental topic in combinatorics on words is the study of repetitions. A {\it repetition} in a word is a factor that is formed by the concatenation of two or more identical blocks. The simplest kind of repetition is a {\it square},
that is, the concatenation of two copies of the same block, such as 
the English word {\tt hotshots}.  A famous  conjecture of Fraenkel and Simpson \cite{FS98}  states that a word of length $n$ contains fewer than $n$  distinct square factors. Experiments strongly suggest that the conjecture is true, but a theoretical proof of the conjecture seems difficult. In \cite{FS98}, the authors proved a bound of $2n$. In \cite{I07}, Ilie improved this bound to $2n-\Theta(\log n)$, and recently Deza et al.~showed the current best bound of $\frac{11}{6}n$ \cite{DeFrTh15}, but the conjectured bound is still out of reach.

Other variations on counting squares include counting squares in
partial words (e.g., \cite{BSMS09}) and pseudo-repetitions
(e.g., \cite{GMMNT13}).

Among the different generalizations of the notion of repetition, a prominent one is that of an abelian repetition. An {\it abelian repetition} in a word is a factor that is formed by the concatenation of two or more blocks that have the same number of occurrences of each letter in the alphabet. Of course, the simplest kind of abelian repetition is an \emph{abelian square}, that is, the concatenation of a word with an anagram of itself, such as the English word
{\tt intestines}.
Abelian squares were considered in 1961 by Erd\H{o}s \cite{Erdos61},
who conjectured that there exist infinite words avoiding abelian squares.
This conjecture was later confirmed,
and the smallest possible size of an alphabet for 
which it holds is known to be $4$ \cite{Ker92}.

We focus on the maximum number of distinct abelian squares that a word can contain.
In contrast to the case of ordinary squares, a word of length $n$ can contain $\Theta(n^2)$ distinct abelian-square factors (see \cite{Ry14}). Since the total number of factors in a word of length $n$ is quadratic in $n$, this means that there exist words in which a constant fraction of all factors are abelian squares. So we turn our attention to infinite words, and we ask whether there exist infinite words such that for every $n$ the factors of length $n$ contain, on average,
a number of distinct abelian-square factors that is quadratic in $n$.   We call such an infinite word \emph{abelian-square-rich}. Since a random binary word of length $n$ contains $\Theta(n\sqrt{n})$  distinct abelian-square factors \cite{Ch14}, the existence of abelian-square-rich words is not immediate. We also introduce \emph{uniformly abelian-square-rich} words; these are infinite words such that for every $n$, every factor of length $n$ contains a quadratic number of distinct abelian squares. Of course, if a word is uniformly abelian-square-rich, then it is abelian-square-rich, but the converse is not true in general --- we provide in this paper an example of a  word that is abelian-square-rich but not uniformly abelian-square-rich. However, we show that for linearly recurrent words the two definitions are equivalent. Moreover, we prove that if an infinite word $w$ is uniformly abelian-square-rich, then $w$ has bounded exponent (that is, there exists an integer $k\geq 2$ such that $w$ does not contain any repetition of order $k$ as a factor).

We then prove that the famous Thue-Morse word is uniformly abelian-square-rich. Furthermore, we look at the function that counts the number of distinct abelian squares of length $2n$ in the Thue-Morse word and prove that this function is $2$-regular.

Then we look at the class of Sturmian words; these are aperiodic infinite words with the lowest possible factor complexity. In this case, we prove that a Sturmian word has bounded exponent if and only if it is uniformly abelian-square-rich, and leave open the question of determining whether a Sturmian word is not abelian-square-rich in the case when it does not have bounded exponent.

\section{Notation and Background}

Let $\Sigma=\{a_1,a_2,\ldots,a_{\sigma}\}$ be an ordered $\sigma$-letter alphabet. Let $\Sigma^{*}$ stand for the free monoid generated by $\Sigma$, whose elements are called \emph{words} over $\Sigma$. The \emph{length} of a word $w$ is denoted by $|w|$. The \emph{empty word}, denoted by $\epsilon$, is the unique word of length zero and is the neutral element of $\Sigma^{*}$. We also define $\Sigma^{+}=\Sigma^{*}\setminus \{\epsilon\}$.
 
A \emph{prefix} (respectively, a \emph{suffix}) of a word $w$ is a word $u$ such that $w=uz$ (respectively, $w=zu$) for some word $z$. A \emph{factor} of $w$ is a prefix of a suffix (or, equivalently, a suffix of a prefix) of $w$.  
The set of prefixes, suffixes and factors of the word $w$ are denoted,
respectively, by $\Pref(w)$, $\Suff(w)$ and $\Fact(w)$.
From the definitions, we have that $\epsilon$ is a prefix, a suffix and a factor of every word. 

A word $w$ is a \emph{$k$-power} (also called a \emph{repetition of order $k$}), for an integer $k\geq 2$, if there exists a nonempty word $u$ such that $w=u^k$. A $2$-power is called a \emph{square}. The \emph{period} of a word $w=w_1w_2\cdots w_{|w|}$ is the minimal integer $p$ such that $w_{i+p}=w_i$ for every $1\leq i\leq |w|-p$. The \emph{exponent} $e(w)$ of a word $w$ is the ratio between its length $|w|$ and its period $p$. For example, the period of $w=abaab$ is $p=3$, hence $e(w)=5/3$. Of course, if a word $w$ avoids $k$-powers (that is, no factor of $w$ is a $k$-power), then the supremum of the exponents of factors of $w$ is smaller than $k$.

For a word $w$ and a letter $a_i\in \Sigma$, we let $|w|_{a_i}$ denote the number of occurrences of $a_i$ in $w$. The \emph{Parikh vector} (sometimes called the \emph{composition vector}) of a word $w$ over $\Sigma=\{a_1,a_2,\ldots,a_{\sigma}\}$ is the vector $P(w)=(|w|_{a_1},|w|_{a_2},\ldots,|w|_{a_{\sigma}})$. 
An \emph{abelian $k$-power} is a nonempty word of the form $v_1v_2\cdots v_k$ where all the $v_i$  have the same Parikh vector (and therefore in particular the same length).
An abelian $2$-power is called an \emph{abelian square}; an example in English is the word {\tt reappear}.

An \emph{infinite word} $w$ over $\Sigma$ is an infinite sequence of letters from $\Sigma$, that is, a function $w:\N \mapsto \Sigma$. 
An infinite word is \emph{recurrent} if each of its factors occurs infinitely often.
Given an infinite word $w$, the \emph{recurrence index} $R_w(n)$ of $w$ is defined to be the least integer $m$ such that every factor of $w$ of length $m$ contains all factors of $w$ of length $n$, or $+\infty$ if such an integer does not exist. If the recurrence index is finite for every $n$, the infinite word $w$ is called \emph{uniformly recurrent} and the function $R_w(n)$ the \emph{recurrence function} of $w$. A uniformly recurrent word is of course recurrent, but the converse is not always true. For example, the \emph{Champernowne word} $w=011011100101\cdots$, obtained by concatenating the base-2 representations of the natural numbers, is recurrent but not uniformly recurrent (to see this, it is sufficient to observe that it contains arbitrarily large consecutive blocks of the same letter).
A uniformly recurrent word  $w$ is called  \emph{linearly recurrent} if the ratio $R_w(n)/n$ is bounded by a constant. Given a linearly recurrent word $w$, the real number $r_w=\limsup_{n\to\infty}R_w(n)/n$ is called the \emph{recurrence quotient} of $w$. 
The \emph{factor complexity function} (sometimes called {\it subword 
complexity}) of an infinite word $w$ is the integer function $p_w(n)$ defined by $p_w(n)=|\Fact(w)\cap \Sigma^n|$. An infinite word $w$ has \emph{linear complexity} if $p_w(n)=O(n)$. In particular, if a word is linearly recurrent, then it has linear complexity (see, for example, \cite{Dur99}).

A \emph{substitution} over the alphabet $\Sigma$ is a map $\tau:\Sigma \mapsto \Sigma^+$. A substitution $\tau$ over $\Sigma$ can be naturally extended to a (non-erasing) morphism from $\Sigma^*$ to $\Sigma^*$.
A substitution can be iterated: for every substitution $\tau$ and every $n>0$, using the extension to a morphism, one can define the substitution $\tau^n$. A substitution $\tau$ is \emph{$r$-uniform} if there exists an integer $r\geq 1$ such that for all $a\in \Sigma$, $|\tau(a)|=r$. A substitution is called \emph{uniform} if it is $r$-uniform for some $r\geq 1$. A substitution $\tau$ is \emph{primitive} if there exists an integer $n\geq 1$ such that for every $a\in \Sigma$, the word $\tau^n(a)$ contains every letter of $\Sigma$ at least once. In this paper, we will only consider primitive substitutions such that $\tau(a_1)=a_1v$ for a letter $a_1$ and some nonempty word $v$. These substitutions always have a fixed point, which is the infinite word $w=\lim_{n\to\infty}\tau^n(a_1)$. Moreover, this fixed point is linearly recurrent (see, for example, \cite{DZ00}) and therefore has linear complexity.

For an integer $k\geq 2$, we say that an infinite word $w$ is \emph{$k$-power-free} if no factor of $w$ is a $k$-power. If an infinite word $w$ is $k$-power-free for some $k$, we say that $w$ has \emph{bounded exponent}. 
For example, if a word is linearly recurrent, then it has bounded exponent (see, for example, \cite{Dur99}).

\section{Abelian-Square-Rich Words}

Kociumaka et al.~\cite{Ry14} showed that a word of length $n$ can contain a number of distinct abelian-square factors that is quadratic in $n$. For the sake of completeness, we give a proof of this fact here.

\begin{proposition}\label{prop:as}
 A word of length $n$ can contain $\Theta(n^{2})$ distinct abelian-square factors.
\end{proposition}

\begin{proof}
 Consider the word $w_n=a^{n}ba^{n}ba^{n}$, of length $3n+2$. For every $0\leq i,j\leq n$ such that $i+j+n$ is even, the factor $a^{i}ba^{n}ba^{j}$ of $w$ is an abelian square. Since the number of possible choices for the pair $(i,j)$ is quadratic in $n$, we are done. 
\end{proof}

Motivated by the previous result, we might ask whether there exist infinite words such that all their factors contain a number of distinct abelian squares that is quadratic in their length. But first we relax this condition and consider words in which, for every $n$, each factor of length $n$ contains, on average, a number of distinct abelian-square factors that is quadratic in $n$.

Let us define some notation. Given a finite or infinite word $w$, we
let $\ass{w} {n}$ denote the number of distinct abelian-square factors of
$w$ of length $n$. Of course,  $\ass{w} {n} =0$ if $n$ is odd, so this
quantity is significant only for even values of $n$. Furthermore, for a
finite word $w$ of length $n$, we let $\as{w} =\sum_{m\leq n} \ass{w} {m}$ denote the
total number of distinct abelian-square factors, of all lengths, in $w$.

\begin{definition}
An infinite word $w$ is \emph{abelian-square-rich} if there exists a positive constant $C$ such that for every $n$ one has 
\[\frac{1}{p_w(n)}\sum_{v\in \Fact(w)\cap\Sigma^n} \as{v} \geq C n^2.
\]
\end{definition}

Notice that Christodoulakis et al. \cite{Ch14} proved that a binary word of length $n$ contains  $\Theta(n\sqrt{n})$  distinct abelian-square factors on average; 
hence a random infinite binary word is almost surely not  abelian-square-rich.

In an abelian-square-rich word the number of distinct abelian squares contained in any factor is, on average, quadratic in the length of the factor. A stronger condition is that \emph{every} factor contains a quadratic number of distinct abelian squares. We thus 
introduce the concept of uniformly abelian-square-rich words.

\begin{definition}
An infinite word $w$ is {\it uniformly abelian-square-rich\/} if 
there exists a positive constant $C$ such that $\as{v} \geq C |v|^2$
for all $v \in \Fact(w)$.
\end{definition}

Clearly, if a word is uniformly abelian-square-rich, then it is also abelian-square-rich, but the converse is not always true (we will provide an example of a  word that is abelian-square-rich but not uniformly abelian-square-rich at the end of this section). However, in the case of linearly recurrent words, the two definitions are equivalent, as shown in the next lemma.

\begin{lemma}\label{lem:un}
Let $w$ be an infinite word. If $w$ is abelian-square-rich and linearly recurrent, then it is uniformly abelian-square-rich.
\end{lemma}

\begin{proof}
If $w$ is linearly recurrent, then there exists a positive integer $K$ such that, for every $n$, every factor of $w$ of length $Kn$ contains all the factors of $w$ of length $n$. Let $v$ be a factor of $w$ of length $n$ containing the largest number of distinct abelian squares among the factors of $w$ of length $n$. Hence the number of distinct abelian squares in $v$ is at least the average number of distinct abelian squares in a factor of $w$ of length $n$. Since $w$ is abelian-square-rich, the number of distinct abelian squares in $v$ is greater than or equal to $C'n^2$, for a positive constant $C'$. Since $v$ is contained in every factor of $w$ of length $Kn$, the number of distinct abelian squares in every factor of $w$ of length $Kn$ is greater than or equal to $C'n^2$, whence $w$ is uniformly abelian-square-rich.
\end{proof}

The following lemma will be useful in the next sections.

\begin{lemma}\label{lem:lin}
 Let $w$ be a linearly recurrent infinite word. If there exists a positive constant $C$ such that for every 
 $n$ one has $\sum_{m\leq n}\ass{w} {m}\geq Cn^2$, then $w$ is (uniformly) abelian-square-rich. 
\end{lemma}

\begin{proof}
Since $w$ is linearly recurrent,  there exists a constant $K$ such that, for every $n$, every factor of $w$ of length $Kn$ contains all the factors of $w$ of length $n$. Take any factor $u$ of $w$ of length $Kn$. Since $u$ contains all the factors of $w$ of length $n$, it contains a number of distinct abelian-square factors that is larger than a constant times $n^2$. Therefore, every factor of length $n$ (and hence in particular, on average) contains a number of distinct abelian-square factors that is proportional to the square of its length.
\end{proof}

In the next proposition we show that a necessary condition for a word to be uniformly abelian-square-rich is that it does not contain factors with arbitrarily large exponent.

\begin{proposition}\label{prop:new}
Let $w$ be an infinite word. If $w$ is uniformly abelian-square-rich, then $w$ has bounded exponent.
\end{proposition}

\begin{proof} 
Let $u^k$ be a nonempty factor of $w$ of length $n=km$, $m=|u|$. Every abelian square in $u^k$ has an occurrence starting at a position smaller than $m$. We separate the abelian squares of $u^k$ in two disjoint sets: those whose first occurrence ends at, or before, position $m$ (i.e., those occurring in $u$), and those whose first occurrence ends after position $m$. Since there are no more than $m^2$ distinct abelian squares of the first kind and no more than $m\cdot (k-1)m=(k-1)m^2$ distinct abelian squares of the second kind, we have that $u^k$ contains no more than $km^2=n^2/k$ distinct abelian-square factors. 
 
If $w$ does not have bounded exponent then, for every $k\geq 2$, $w$ contains a nonempty factor $v$ of the form $v=u^k$, for some $u$. Hence, for every positive constant $C$, taking $k$ such that $1/k<C$, the word $w$ contains a factor $v$ such that $\as{v} < C |v|^2$.
\end{proof}


To conclude this section, we exhibit an example of a word that is abelian-square-rich but not uniformly abelian-square-rich.

Consider the sequence of words $(w_k)_{k\geq 1}$  defined by: $w_1=aabaabaab$,  and for every $k> 1$ 
\begin{align}\label{eq:len}
w_{k} &= w_{k-1}a^{2^k}ba^{2^k}ba^{2^k}b.
\end{align}

\begin{proposition}
 The infinite word $w=\lim_{k\to \infty}w_k$ is  abelian-square-rich but not uniformly abelian-square-rich.
\end{proposition}

\begin{proof}
Let us first prove that $w$ is abelian-square-rich. We first observe that $w$ has linear complexity. Indeed, this is an immediate consequence of the fact that for every $n$ there is a constant number of distinct factors of length $n$ that can be extended to the right both by the letter $a$ or by the letter $b$ to factors of $w$ of length $n+1$ (these factors are usually called \emph{right special factors}). Actually, for every $n$ the number of right special factors of $w$ of length $n$ is bounded by $4$, since these can only be of the following kinds:
\begin{enumerate}
\item $a^n$;
\item $a^iba^j$, for some $i\geq 0$ and $j>0$ such that $j=2^r$ for some $r$ and $i+j+1=n$;
\item $a^iba^jba^j$, for some $i\geq 0$ and $j>0$ such that $j=2^r$ for some $r$ and $i+2j+2=n$;
\item $a^iba^jba^jba^j$, for some $i\geq 0$ and $j>0$ such that $j=2^r$ for some $r$ and $i+3j+3=n$.
\end{enumerate}
It is readily verified that if a factor of $w$ contains more than $3$ $b$'s or does not end in $a^{2^r}$ for some $r>0$, then it cannot be a right special factor of $w$.

So, it is sufficient to prove that for every $n$ there is a linear (in $n$) number of factors of $w$ of length $n$ that contain a number of distinct abelian-square factors proportional to $n^2$. This follows from the fact that for every $n$ there is a linear number of distinct factors that contain $a^{2^k}ba^{2^k}ba^{2^k}$ as factor, for some value of $2^k$ proportional to $n$ (for example, take the largest $k$ such that $2^k\leq n/100$, and consider the factors of length $n$ in which $a^{2^k}ba^{2^k}ba^{2^k}$ appears at different positions), and these factors contain a number of distinct abelian-square factors proportional to the square of their length (see the proof of Proposition \ref{prop:as}), hence proportional to $n^2$.

Finally, $w$ is not  uniformly abelian-square-rich  by Proposition \ref{prop:new}, since it contains arbitrarily large powers of the letter~$a$.
\end{proof}

It remains to prove that a  recurrent (or even a uniformly recurrent) word exists that is abelian-square-rich but not uniformly abelian-square-rich, but such an example is probably more technical and involved.

\section{The Thue-Morse Word}

Let 
\[
{\bf t} =011010011001011010010110\cdots
\]
be the Thue-Morse word, i.e., the fixed point starting with $0$ of the uniform substitution $\mu:0\mapsto 01,1\mapsto 10$. It is well known that $\bf t$ is linearly recurrent and that $\bf t$ does not contain any factor with exponent larger than $2$. In particular, $\bf t$ does not contain \emph{overlaps}, i.e., factors of the form $avava$, with $a\in \{0,1\}$ and $v\in \{0,1\}^*$.

For every $n\geq 4$, the factors of length $n$ of $\bf t$ belong to two disjoint sets: those that start only at even positions in $\bf t$, and those that start only at odd positions in $\bf t$. 
This is a consequence of two facts:  first, that $\bf t$ is overlap-free  (and so $0101$ cannot be preceded by $1$ nor followed by $0$) and second, that $00$ and $11$ are not images under $\mu$ of letters, so they cannot appear at even positions.

Let $p(n)$ be the factor complexity function of $\bf t$. It is known \cite[Proposition 4.3]{Br89} that for every $n\geq 1$ one has 
\begin{equation}\label{eq:tm}
 p(2n)=p(n)+p(n+1), \hspace{6mm} p(2n+1)=2p(n+1).
\end{equation}

We define  $f_{aa}(n)$ (respectively, $f_{ab}(n)$) to be the number of factors of $\bf t$ of length $n$ that begin and end with the same letter (respectively, with different letters).
The next lemma (proved in \cite{CaFiScZa15}) shows that the Thue-Morse word has the property that for every $n$, at least one-third of the length-$n$ factors begin and end with the same letter, and at least one-third of the length-$n$ factors begin and end with different letters. 

\begin{lemma}[\protect{\kern-0.3em}\cite{CaFiScZa15}]
\label{lem:third}
For every $n\geq 2$,
one has $f_{aa}(n)\geq  p(n)/3$ and $f_{ab}(n)\geq  p(n)/3$.
\end{lemma}

Since from \eqref{eq:tm} we have $p(n)\geq 3(n-1)$ for every $n$, we get the following result.

\begin{corollary}\label{cor:TM}
  For every $n\geq 2$, one has $f_{aa}(n)\geq  n-1$ and $f_{ab}(n)\geq  n-1$.
\end{corollary}

We are now ready to prove that the Thue-Morse word is uniformly abelian-square-rich.

\begin{proposition}
 The Thue-Morse word $\bf t$ is uniformly abelian-square-rich.
\end{proposition}

\begin{proof}
 Let $u$ be a factor of length $n>1$ of $\bf t$ that begins and ends with the same letter. Since the image of every even-length word under $\mu$ is an abelian square, we have that $\mu^2(u)$ is an abelian-square factor of $\bf t$ of length $4n$ that begins and ends with the same letter. Moreover, the word obtained from $\mu^2(u)$ by removing the first and the last letter  is an abelian-square factor of $\bf t$ of length $4n-2$. So, by Corollary \ref{cor:TM}, $\bf t$ contains at least $n-1$ distinct abelian-square factors of length $4n$ and at least $n-1$ distinct abelian-square factors of length $4n-2$. This implies that for every even $n$ the number of distinct abelian-square factors of $\bf t$ of length $n$ is linear in $n$. Hence, for every $n$ the number of distinct abelian-square factors of $\bf t$ of length at most $n$ is quadratic in $n$. The statement then follows from Lemmas \ref{lem:lin} and \ref{lem:un}.
\end{proof}

\subsection{More detailed analysis for the Thue-Morse word}

Let $f(n)$ denote the number of distinct abelian squares of length
$2n$ (or, equivalently, of order $n$) in the Thue-Morse sequence $\bf t$.
Table~\ref{tabb1} gives the first few terms of this sequence.

\medskip

\begin{table}[H]
\begin{center}
\begin{tabular}{|c|cccccccccccccccccccc}
$n$ & 1& 2& 3& 4& 5& 6& 7& 8& 9&10&11&12&13&14&15&16&17&18&19&20 \\
\hline
$f(n)$ & 2& 4& 4&10& 8&24&10&22&12&36&20&52&24&54&20&46&24&72&32&76 \\
\end{tabular}
\end{center}
\caption{\label{tabb1} First few values of the number of distinct abelian-square factors of order $n$ in the Thue-Morse word.}
\end{table}

\medskip

\noindent We can use the decision procedure discussed in
\cite{Allouche&Rampersad&Shallit:2009,Charlier&Rampersad&Shallit:2012,Goc&Henshall&Shallit:2013,Goc&Mousavi&Shallit:2013,Shallit:2013} to analyze the function $f$
in more detail.
Although in general abelian questions about arbitrary
automatic sequences are {\it not\/} decidable using this method,
the Thue-Morse sequence has symmetries that make it
amenable.  One nice feature of this approach is that it can be almost
entirely automated, using the freely-available {\tt Walnut} package \cite{Mousavi:2016}.

Define
$$d(n) := |{\bf t}[0..n-1]|_0 - {n\over 2} $$
to be the {\it prefix defect function\/} for the Thue-Morse sequence,
that is, the function that counts
the number of $0$'s in a prefix of length $n$ over the amount
``expected to occur''.   (Since ${\bf t} \in (01+10)^\omega$,
we expect $n \over 2$ $0$'s in a prefix of length $n$.)

Recall that we say that an infinite word ${\bf a} = (a_n)_{n \geq 0} \in \Delta^\omega$ is
a {\it $k$-automatic sequence} if there exists a deterministic finite automaton with
output (DFAO) $M = (Q, \Sigma, \delta, q_0, \Delta, \tau)$ that computes $\bf a$
in the following sense:  if the input to $M$ is $(n)_k$, the base-$k$ representation of
$n$, then the output $\tau(\delta(q_0, x))$ is equal to $a_n$.  Here $\tau:Q 
\rightarrow \Delta$.  
Without loss of generality, we assume all automata discussed in this section
take, as input, the base-$2$ representations of numbers starting with the most
significant digit.
See, for example, \cite{Allouche&Shallit:2003}.

Now it is not hard to see that ${\bf d} = 
(d_n)_{n \geq 0}$ is a $2$-automatic sequence.  In fact,
the four-state automaton in Figure~\ref{tmdisc} computes it.  Here the notation
$a/b$ in a state indicates that the state name is $a$ and its output is $b$.
Thus, for example, starting in state $0$, and reading 
the base-$2$ expansion of 23 (namely, $10111$), we reach the state
$3$ with output $-{1 \over 2}$.  And indeed, ${\bf t}[0..22]$ has 11 zeroes,
so its defect is $11-{{23} \over 2} = -{1 \over 2}$.

\begin{figure}[H]
\centering
\includegraphics[scale=0.5]{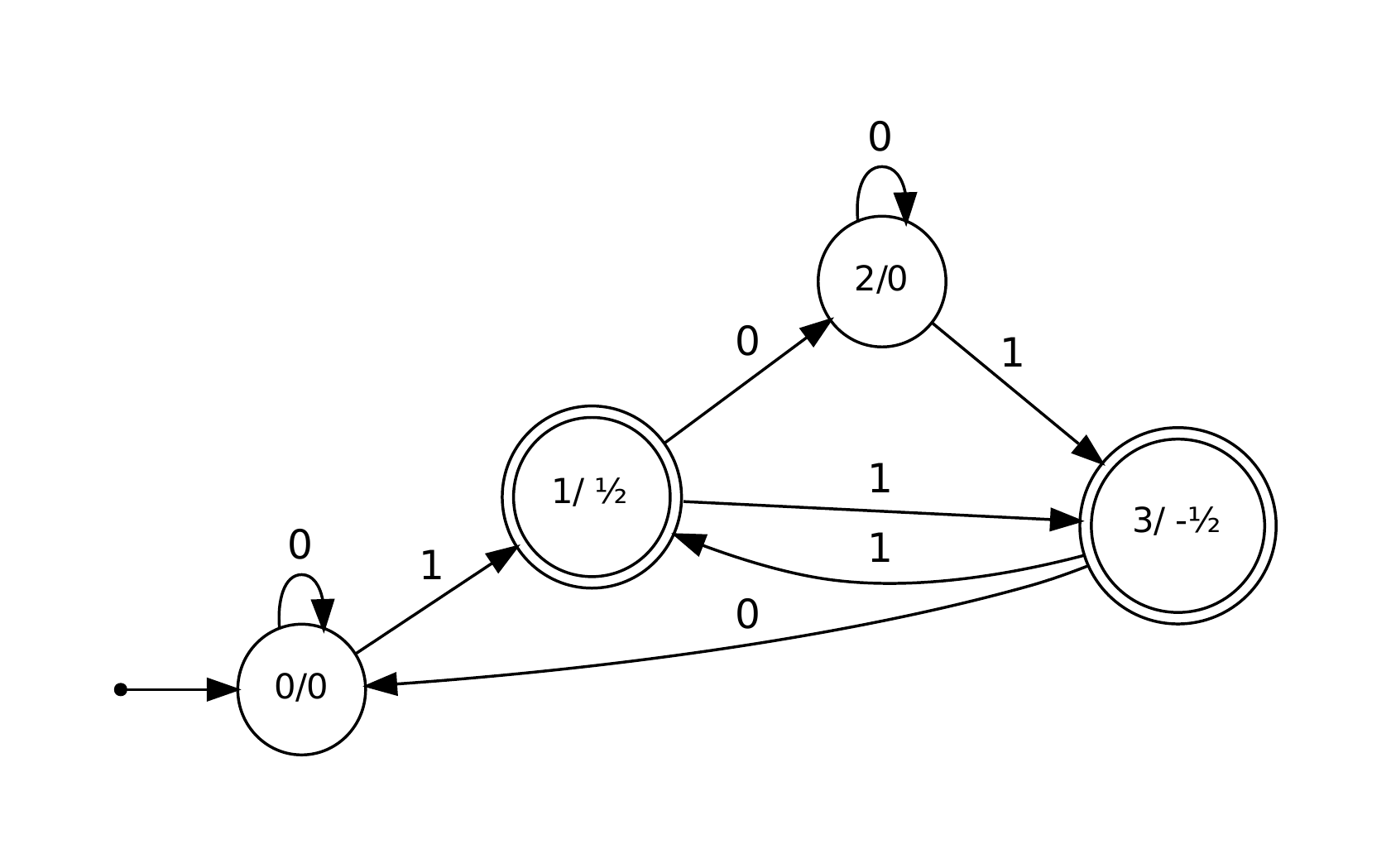} 
\caption{Automaton calculating prefix defect function for {\bf t}.}
\label{tmdisc}
\end{figure}

Now that we have the sequence $\bf d$, 
we define the function $D(i,n)$ be the defect
associated with the factor ${\bf t}[i..i+n-1]$, that is,
$$
D(i,n)  := |{\bf t}[i..i+n-1]|_0 - {n \over 2}.$$
Then we have $D(i,n) = 
d(i+n) - d(i) \in \lbrace -1, -{1\over 2}, 0, {1 \over 2}, 1 \rbrace$.

We can view $D$ as a {\it two-dimensional automatic sequence} 
(see, e.g., \cite[Chapter 14]{Allouche&Shallit:2003}).  Here the input
alphabet is $\lbrace 0,1\rbrace^2$, where the first components of the
input spell out the base-$2$ representation of $i$ and
the second components spell out the base-$2$ representation of $n$.
We write this input as $(i,n)_2$.
(The shorter of the two representations is, if necessary, padded with
leading zeroes so that the two representations can be read in parallel.)
From the DFAO for $\bf d$ we can easily generate an automaton
for the two-dimensional infinite array ${\bf D} = (D(i,n))_{i, n \geq 0}$.
More usefully, 
we can produce $5$ different automata
$A_1, A_2, A_3, A_4, A_5$ accepting
those inputs $(i,n)_2$ for which $D(i,n) =
-1$ (respectively, $-{1\over 2}, 0, {1 \over 2}, 1$). We can think of each
automaton $A_m$, $1 \leq m \leq 5$,
as computing the function ${\mathcal A}_m(i,n)$ that is {\tt true} when
$(i,n)_2$ is
accepted and {\tt false} otherwise.
When we compute
these using {\tt Walnut} we discover that they have 10, 13, 7, 13, 10
states respectively.

Once we have the automata 
$A_1, A_2, A_3, A_4, A_5$, we can create an automaton $T$ accepting
those $(i,n)_2$ such that ${\bf t}[i..i+2n-1]$ is an abelian square.  This
automaton is created by taking the disjunction
of the assertions ${\mathcal A}_m(i,n) \wedge {\mathcal A}_m(i+n,n)$ for $1 \leq m \leq 5$,
and is implemented by an automaton with 36 states.

Finally, using $T$ we can create an automaton $U$ accepting
those $(i,n)_2$ such that 
${\bf t}[i..i+2n-1]$ is a novel abelian square; that is,
an abelian square of length $2n$ that has never appeared previously
in $\bf t$:
$$U(i,n) := T(i,n) \wedge (\forall j\  (j<i) \implies (\exists k\ (k<2n) 
\wedge ({\bf t}[i+k]\not= {\bf t}[j+k]))) .$$
The automaton $U$ has 64 states, and can be computed using the techniques
described in \cite{Goc&Henshall&Shallit:2013}.  Again, we can view
$U$ as computing the function
${\mathcal U}(i,n)$ that is {\tt true} when
$(i,n)_2$ is accepted and {\tt false} otherwise.

Now define $f(n) = | \{ i \ : \ {\mathcal U}(i,n) \} |$.  As a moment's reflection
will reveal, this is the number of distinct abelian squares in $\bf t$
of order $n$ (and length $2n$).   As the techniques in 
\cite{Charlier&Rampersad&Shallit:2012} show, from the transition diagram of
the automaton $U$ we can immediately deduce a so-called {\it linear representation}
for the function $f$,
that is, square matrices
$M_0, M_1$ and vectors $v, w$ such that
$$ f(n) = v M_{a_1} \cdots M_{a_i} w $$
if $(n)_2 = a_1 \cdots a_i$.  
We can minimize the representation using the algorithm in \cite[Section~2.3]{Berstel&Reutenauer:2011}, obtaining
the following linear representation of rank $11$:

\begin{align*}
v & := [1, 0, 0, 0, 0, 0, 0, 0, 0, 0, 0] \\
\ \\
M_0 &:= \left[ \begin {array}{ccccccccccc} 1&0&0&0&0&0&0&0&0&0&0\\ \noalign{\medskip}0&0&1&0
&0&0&0&0&0&0&0\\ \noalign{\medskip}0&0&0&0&1&0&0&0&0&0&0\\ \noalign{\medskip}0&0&0&0
&0&0&1&0&0&0&0\\ \noalign{\medskip}0&0&0&0&0&0&0&0&1&0&0\\ \noalign{\medskip}0&0&0&0
&0&0&0&0&0&0&1\\ \noalign{\medskip}0&0&{\frac {64}{13}}&1&-{\frac {171}{26}}&-{
\frac {69}{26}}&0&0&{\frac {43}{26}}&{\frac {9}{13}}&{\frac {51}{26}}
\\ \noalign{\medskip}0&0&{\frac {90}{13}}&-1&-{\frac {275}{26}}&{\frac {9}{26}}&0&2&
{\frac {69}{26}}&-{\frac {17}{13}}&{\frac {51}{26}}\\ \noalign{\medskip}0&0&-2&0&1&0
&0&0&2&0&0\\ \noalign{\medskip}0&0&{\frac {68}{13}}&0&-{\frac {88}{13}}&-{\frac {33}
{13}}&0&0&{\frac {20}{13}}&{\frac {12}{13}}&{\frac {34}{13}}\\ \noalign{\medskip}0&0
&{\frac {54}{13}}&0&-{\frac {89}{13}}&-{\frac {22}{13}}&0&0&{\frac {35}{13}}&{\frac 
{8}{13}}&{\frac {27}{13}}\end {array} \right] 
\end{align*}

\begin{align*}
M_1 & := \left[ \begin {array}{ccccccccccc} 0&1&0&0&0&0&0&0&0&0&0\\ \noalign{\medskip}0&0&0&
1&0&0&0&0&0&0&0\\ \noalign{\medskip}0&0&0&0&0&1&0&0&0&0&0\\ \noalign{\medskip}0&0&0&0
&0&0&0&1&0&0&0\\ \noalign{\medskip}0&0&0&0&0&0&0&0&0&1&0\\ \noalign{\medskip}0&0&-{
\frac {8}{13}}&0&-{\frac {8}{13}}&{\frac {10}{13}}&0&0&{\frac {16}{13}}&{\frac {7}{
13}}&-{\frac {4}{13}}\\ \noalign{\medskip}0&0&-{\frac {87}{13}}&-1&{\frac {203}{26}}
&{\frac {3}{26}}&2&0&-{\frac {3}{26}}&{\frac {16}{13}}&-{\frac {61}{26}}
\\ \noalign{\medskip}0&0&-{\frac {7}{13}}&1&{\frac {25}{26}}&-{\frac {93}{26}}&0&0&{
\frac {15}{26}}&{\frac {37}{13}}&-{\frac {7}{26}}\\ \noalign{\medskip}0&0&{\frac {20
}{13}}&0&-{\frac {45}{13}}&-{\frac {12}{13}}&0&0&{\frac {25}{13}}&{\frac {28}{13}}&-
{3\over{13}}\\ \noalign{\medskip}0&0&-{\frac {40}{13}}&0&{\frac {25}{13}}&-{2\over{13}}&0&0&{\frac {
15}{13}}&{\frac {22}{13}}&-{\frac {7}{13}}\\ \noalign{\medskip}0&0&-{\frac {36}{13}}
&0&{\frac {29}{13}}&-{\frac {20}{13}}&0&0&{\frac {7}{13}}&{\frac {38}{13}}&-{\frac {
5}{13}}\end {array} \right]  \\
\ \\
w & := [1,2,4,4,10,8,24,10,22,12,36]^T
\end{align*}

We have therefore proved

\begin{theorem}
The function $f$ counting the number of distinct
abelian squares of order $n$
of $\bf t$ is given by 
$f(n) = v M_{a_1} \cdots M_{a_i} w $
if $(n)_2 = a_1 \cdots a_i$, for the vectors $v,w$ and matrices
$M_0, M_1$ given above.
\end{theorem}
A sequence that can be computed in this manner is called $2$-regular;
see \cite{Allouche&Shallit:1992}.

From the matrices above we can, using the method described in
\cite{Goc&Mousavi&Shallit:2013}, obtain defining recursive relations
for $f$:

\begin{align*}
f(8n) &= -2  f(2n) + 3 f(4n) \\
f(8n+6) &= 4  f(2n+1) -  f(4n+1) + 4 f(4n+2) + f(4n+3) + f(8n+1) - f(8n+2) +  f(8n+3) \\
& \quad - 2 f(8n+4) + 2 f(8n+5) \\
f(16n+3) &= f(2n) - {3 \over 2}  f(2n+1) - f(4n) - {{15} \over 4} f(4n+1) +  5 f(4n+2) +  {5 \over 2} f(8n+1) \\
& \quad + {{13} \over 4}   f(8n+3) - {9 \over 4}  f(8n+4) + {1 \over 2}    f(16n+1) - {1 \over 2}  f(16n+2) \\
f(16n+4) &= f(2n) -{3 \over 2}  f(2n+1) - f(4n) -{7 \over 4}  f(4n+1) +  f(4n+2) + {3 \over 2}  f(8n+1)  \\
& \quad + 2  f(8n+2) + {5 \over 4}  f(8n+3) - {1 \over 4} f(8n+4) {1 \over 2}  f(16n+1) -{1 \over 2}  f(16n+2) \\
f(16n+5) &=  f(2n) + {1 \over 2}  f(2n+1) - f(4n) -{{11}\over 4} f(4n+1) + 3  f(4n+2) + {5 \over 2}  f(8n+1) \\
& \quad + {9 \over 4}   f(8n+3) -{5 \over 4}  f(8n+4) + {1 \over 2}  f(16n+1) -{1 \over 2}  f(16n+2) \\
f(16n+7) &= 2 f(4n+1) -2 f(4n+2) + f(8n+3) \\
f(16n+9) &= f(4n+1) -2 f(4n+2) + f(4n+3) + f(8n+3) +  f(8n+4) \\
f(16n+11) &= 4  f(2n+1) -  f(4n+1) +2  f(4n+2) +  f(4n+3) -  f(8n+3) - f(8n+4) +2  f(8n+5) \\
f(16n+12) &= 8  f(2n+1) -3  f(4n+1) +10  f(4n+2) +3  f(4n+3) +2  f(8n+1) -2   f(8n+2) \\
& \quad +  f(8n+3) -5   f(8n+4) +5  f(8n+5) \\
f(16n+13) &= 2  f(2n+1) -2  f(4n+1) +4 f(4n+2) +  f(4n+3) -2  f(8n+3)  \\
& \quad -2  f(8n+4) +4  f(8n+5) \\
f(16n+15) &= -2  f(2n+1) +  f(4n+3) +2  f(8n+7) \\
f(32n+1) &= -4  f(2n) +2  f(2n+1) +4  f(4n) +7 f(4n+1) -12  f(4n+2) -  f(8n+1) 
\\
& \quad -5  f(8n+3) +5  f(8n+4) -2  f(16n+1) +2  f(16n+2)  \\
f(32n+2) &= -4  f(2n) +4  f(2n+1) +4  f(4n) +8  f(4n+1) -14  f(4n+2)  - 4  f(8n+1) \\
& \quad +  f(8n+2) -5   f(8n+3) +5  f(8n+4) -3  f(16n+1) +4  f(16n+2) \\
f(32n+10) &= 3  f(2n) -{5 \over 2}    f(2n+1) -3  f(4n) -{{17} \over 4}  f(4n+1) +13  f(4n+2) +{{17} \over 2}  f(8n+1) \\
& \quad -  f(8n+2) +{{35} \over 4}  f(8n+3) -{{19} \over 4}  f(8n+4) +{3 \over 2}  f(16n+1) -{3 \over 2}  f(16n+2) \\
f(32n+17) &= 2  f(4n+1) -4  f(4n+2) + 2  f(8n+3) + 2  f(8n+4) +   f(8n+5) \\
f(32n+18) &= 14  f(2n+1) +4  f(4n+1) -4  f(4n+2) +  f(4n+3) +  f(8n+1) -  f(8n+2) \\
& \quad +6  f(8n+3) +  f(8n+4) +2  f(8n+5) \\
f(32n+26) &= -10  f(2n+1) -6  f(4n+1) +8  f(4n+2) +3  f(4n+3) -  f(8n+1) +  f(8n+2) \\
& \quad -8  f(8n+3) -3  f(8n+4) +5  f(8n+5) +3  f(8n+7) +2  f(16n+10)
\end{align*}

\begin{corollary}
$f(4^n - 1) = (4^{n+1} - 4)/3$ and
$f(3 \cdot 2^n) = 14\cdot 2^n - 4$ for $n \geq 1$. 
\label{cor9}
\end{corollary}

We suspect the values in Corollary~\ref{cor9} are, respectively, the
local minima (maxima) of $f(n)/n$, but we do not have a proof yet.

\section{Sturmian Words}

In this section we fix the alphabet $\Sigma=\{\sa{a,b}\}$.

Recall that a (finite or infinite) word $w$ over $\Sigma$ is
\emph{balanced} if and only if for every pair of
factors $u,v$ of $w$ of the same
length, one has $\left| |u|_{\sa{a}}-|v|_{\sa{a}} \right| \le 1$.

We start with a simple lemma.

\begin{lemma}\label{lem:bal}
 Let $w$ be a finite balanced word over $\Sigma$. Then for every $k>0$,  $P(w) \equiv (0,0)$ (mod $k$) if and only if $w$ is an abelian $k$-power.
\end{lemma}

\begin{proof}
Let $w$ be balanced and $P(w)=(ks,kt)$, for a positive integer $k$ and some  $s,t\geq 0$. Then we can write $w=v_1v_2\cdots v_k$ where each $v_i$ has length $s+t$. Now each $v_i$ must have Parikh vector equal to $(s,t)$, otherwise $w$ would not be balanced, whence the `only if' part of the statement follows. The `if' part is straightforward.
\end{proof}

A binary infinite word is \emph{Sturmian} if and only if it is balanced and aperiodic. Sturmian words are precisely the infinite words having $n+1$ distinct factors of length $n$ for every $n\geq 0$. There are many other equivalent definitions of Sturmian words. A classical reference on Sturmian words is \cite[Chapter 2]{LothAlg}. Let us recall the definition of Sturmian words as codings of a rotation.

We fix the torus $I=\mathbb{R}/\mathbb{Z}=[0,1)$. Given $\delta,\gamma$ in $I$, if $\delta > \gamma$, we use the notation $[\delta,\gamma)$ for the interval $[\delta,1)\cup [0,\gamma)$. Recall that given a real number $\alpha$,  $\lfloor \alpha \rfloor$ is the greatest integer smaller than or equal to $\alpha$, $\lceil \alpha \rceil$ is the least integer greater than or equal to $\alpha$, and $\{\alpha\}=\alpha-\lfloor \alpha \rfloor$ is the fractional part of  $\alpha$. 
Notice that $\{-\alpha\}= 1-\{\alpha\}$ for non-integer $\alpha$. 

Let $\alpha\in I$ be irrational, and $\rho\in I$. 
The Sturmian word $s_{\alpha,\rho}$ (respectively, $s'_{\alpha,\rho}$) of  \emph{angle} $\alpha$ and \emph{initial point} $\rho$ is the infinite word $a_{0}a_{1}a_{2}\cdots$ defined by
$$a_{n} =
\begin{cases}
		\sa{b},  & \text {if } \{ \rho + n\alpha \}\in I_{\sa{b}};\\
		\sa{a},  & \text{if } \{ \rho + n\alpha \}\in I_{\sa{a}}.
\end{cases}
$$
where $I_{\sa{b}}=[0,1-\alpha)$ and $I_{\sa{a}}=[1-\alpha,1)$   (respectively,
$I_{\sa{b}}=(0,1-\alpha]$ and $I_{\sa{a}}=(1-\alpha,1]$).

In other words, take the unit circle and consider a point initially in position $\rho$. Then start rotating this point on the circle (clockwise) by an angle $\alpha$, $2\alpha$, $3\alpha$, etc. For each rotation, take the letter $\sa{a}$ or $\sa{b}$ associated with the interval within which the point falls. The infinite sequence obtained in this way is the Sturmian word $s_{\alpha,\rho}$ (or $s'_{\alpha,\rho}$, depending on the choice of the two intervals). See \figurename~\ref{Fig:gab1} for an illustration.

\begin{figure}[!ht]
\centering
\includegraphics[scale=0.3]{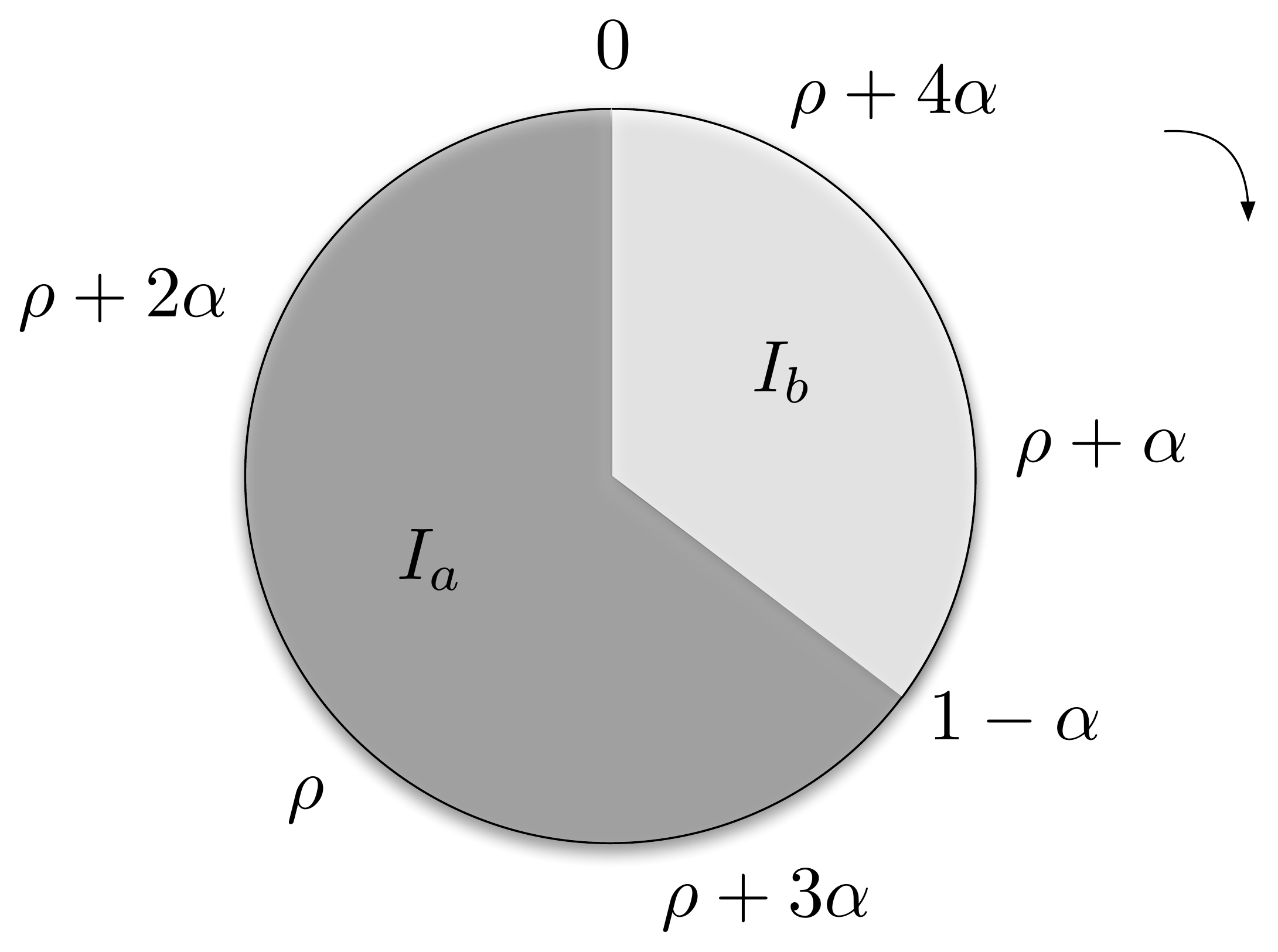} 
\caption{The rotation of angle $\alpha=\phi-1$ (where $\phi=(1+\sqrt 5 )/2\approx 1.618$ is the golden ratio) and initial point $\rho=\alpha$ generating the Fibonacci word $F=s_{\phi-1,\phi-1}=\sa{abaababaabaabab}\cdots$. \label{Fig:gab1}}
\end{figure}

For example, if $\phi=(1+\sqrt 5 )/2\approx 1.618$ is the golden ratio,  the Sturmian word $$F=s_{\phi-1,\phi-1}=\sa{abaababaabaababaababaabaababaabaab}\cdots$$ is called the \emph{Fibonacci word}. 

A Sturmian word for which $\rho=\alpha$, such as the Fibonacci word, is called \emph{characteristic}. Note that for every $\alpha$ one has $s_{\alpha,0}=\sa{b}s_{\alpha,\alpha}$ and $s'_{\alpha,0}=\sa{a}s_{\alpha,\alpha}$.

An equivalent way to visualize the coding of a rotation consists of fixing the point and rotating the intervals.  In this representation, the interval $I_{\sa{b}}=I_{\sa{b}}^{0}$ is rotated at each step, so that after $i$ rotations it is transformed into the interval $I_{\sa{b}}^{-i}=[\{-i\alpha\},\{-(i+1)\alpha\})$, while $I_{\sa{a}}^{-i}=I\setminus I_{\sa{b}}^{-i}$. 

This representation is convenient, since one can read within it not only a Sturmian word, but also all of its factors. More precisely, for every positive integer $n$, the factor of length $n$ of $s_{\alpha,\rho}$ starting at position $j\geq 0$ is determined by the value of $\{\rho+j\alpha\}$ only. Indeed, for every $j$ and $i$, we have
$$a_{j+i} = \begin{cases}
\sa{b},  & \text{if $\{\rho + j\alpha\}\in I_{\sa{b}}^{-i}$;}\\
\sa{a}, & \text{if $\{\rho + j\alpha\}\in I_{\sa{a}}^{-i}$.}
\end{cases}
$$
As a consequence, we have that given a Sturmian word $s_{\alpha,\rho}$ and  a positive integer $n$, the $n+1$ different factors of $s_{\alpha,\rho}$ of length $n$ are completely determined by the intervals $I_{\sa{b}}^{0}, I_{\sa{b}}^{-1},\ldots, I_{\sa{b}}^{-(n-1)}$, that is, only by the points $\{-i\alpha\}$ for $0\leq i< n$. In particular, they do not depend on $\rho$, so that the set of factors of $s_{\alpha,\rho}$ is the same as the set of factors of $s_{\alpha,\rho'}$ for every $\rho$ and $\rho'$.
Hence, from now on, we let $s_{\alpha}$ denote a Sturmian word of angle $\alpha$.

If we arrange the $n+2$ points $0,1,\{-\alpha\},\{-2\alpha\},\ldots,\{-n \alpha\}$ in increasing order, we determine a partition of $I$ in $n+1$ subintervals, $L_0(n),L_{1}(n),\ldots,L_{n}(n)$. Each of these subintervals is in bijection with a different factor of length $n$ of 
 $s_\alpha$ (see \figurename~\ref{Fig:gab3}). 

\begin{figure}[!ht]
\centering
\includegraphics[scale=0.5]{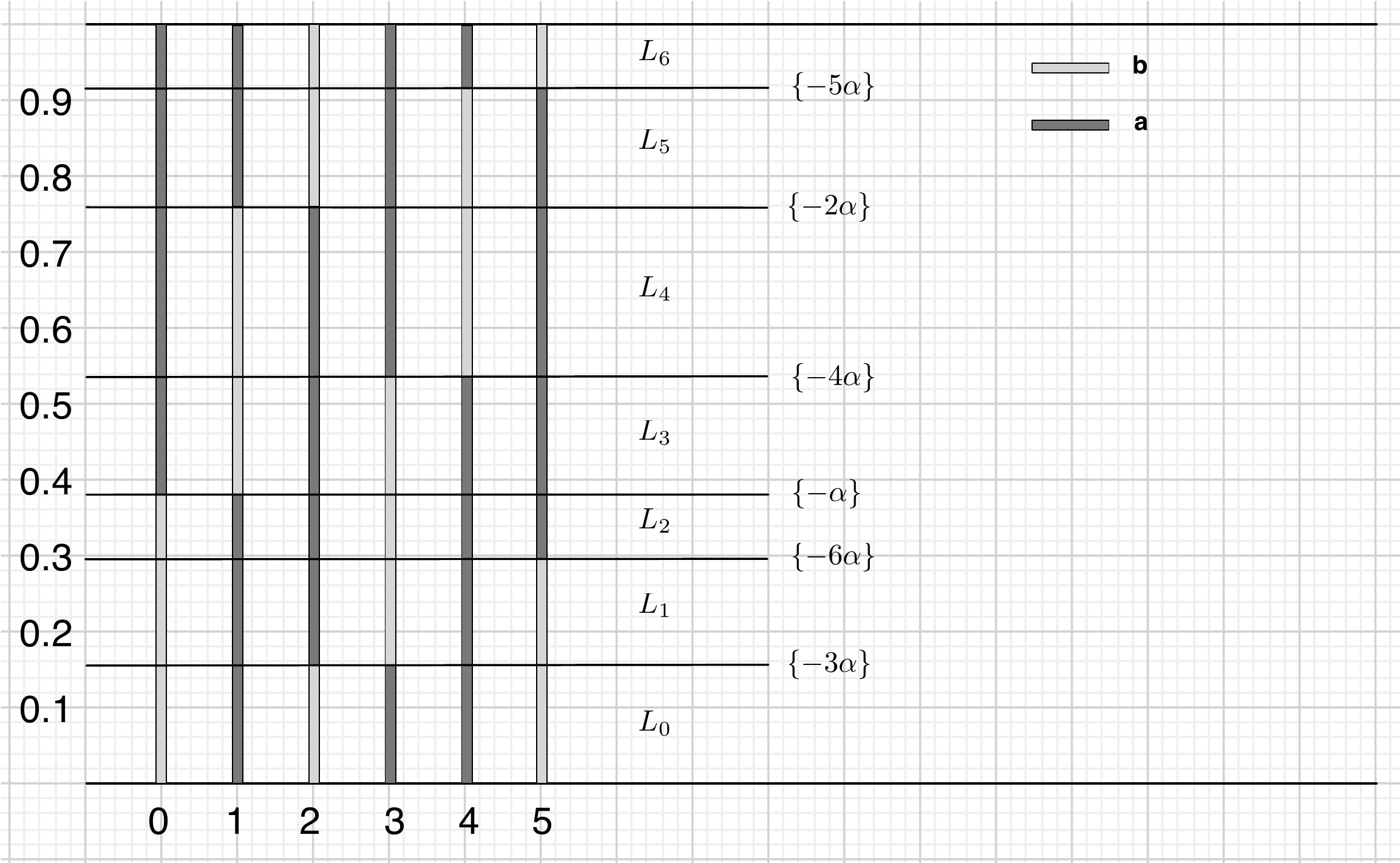} 
\caption{%
The points $0$, $1$ and  $\{-\alpha\}$, $\{-2\alpha\}$, $\{-3\alpha\}$, $\{-4\alpha\}$, $\{-5\alpha\}$, $\{-6\alpha\}$ ($\alpha=\phi-1$),
arranged in increasing order, define the intervals $L_{0}(6)\approx[0,0.146)$, $L_{1}(6)\approx[0.146,0.292)$, $L_{2}(6)\approx[0.292,0.382)$, $L_{3}(6)\approx[0.382,0.528)$, $L_{4}(6)\approx[0.528,0.764)$, $L_{5}(6)\approx[0.764,0.910)$, $L_{6}(6)\approx[0.910,1)$. Each interval is associated with one of the factors of length $6$ of the Fibonacci word, respectively $\sa{babaab},\sa{baabab},\sa{baabaa},\sa{ababaa},\sa{abaaba},\sa{aababa},\sa{aabaab}$. \label{Fig:gab3}}
\end{figure}

Recall that a factor of length $n$ of a Sturmian word  $s_{\alpha}$ has a Parikh vector equal either to $(\lfloor n\alpha \rfloor , n-\lfloor n\alpha \rfloor )$ (in which case it is called \emph{light}) or to $(\lceil n\alpha \rceil , n-\lceil n\alpha \rceil) $ (in which case it is called \emph{heavy}).
The following proposition relates the intervals $L_{i}(n)$ to the Parikh vectors of the factors of length $n$ (see \cite{dlt,tcs16,Rigo13}).

\begin{proposition}\label{pro:main}
Let $s_{\alpha}$ be a Sturmian word of angle $\alpha$, and $n$ a positive integer.
Let $t_{i}$ be the factor of length $n$ associated with the interval $L_{i}(n)$. Then $t_{i}$ is
  heavy if  $L_{i}(n)\subset [\{-n\alpha\},1)$, while it is light if  $L_{i}(n)\subset [0,\{-n\alpha\})$.
\end{proposition}

\begin{example}
 Let $\alpha=\phi-1\approx 0.618$ and $n=6$. We have $6\alpha\approx 3.708$, so that $\{-6\alpha\}\approx 0.292$. The reader can see in \figurename~\ref{Fig:gab3} that the factors of length $6$ 
corresponding to intervals above (respectively, below) $\{-6\alpha\}\approx 0.292$ all have Parikh vector $(4,2)$ (respectively, $(3,3)$). That is, the intervals $L_{0}$ and $L_{1}$ are associated with light factors (\sa{babaab}, \sa{baabab}), while the intervals $L_{2}$ to $L_{6}$ are associated with  heavy factors (\sa{baabaa}, \sa{ababaa}, \sa{abaaba}, \sa{aababa}, \sa{aabaab}). 
\end{example}

Observe that, by Lemma \ref{lem:bal}, every factor of a Sturmian word having even length and containing an even number of $\sa{a}$'s (or, equivalently, of $\sa{b}$'s) is an abelian square. The following proposition relates the abelian-square factors of a Sturmian word of angle $\alpha$  with the arithmetic properties of $\alpha$.

\begin{proposition}\label{prop:hl}
Let $s_{\alpha}$ be a Sturmian word of angle $\alpha$, and $n$ a positive even integer.
Let $t_{i}$ be the factor of length $n$ associated with the interval $L_{i}(n)$. Then $t_{i}$ is
an abelian square if and only if $L_{i}(n)\subset [\{-n\alpha\},1)$ if $\lfloor n \alpha \rfloor$ is even, or $L_{i}(n)\subset [0,\{-n\alpha\})$ if $\lfloor n \alpha \rfloor$ is odd.
\end{proposition}

\begin{proof}
By Proposition \ref{pro:main}, $t_i$ is heavy if $L_{i}(n)\subset [\{-n\alpha\},1)$, while it is light if $L_{i}(n)\subset [0,\{-n\alpha\})$. If $\lfloor n \alpha \rfloor$ is even, then every light factor of length $n$ contains an even number of $\sa{a}$'s and hence is an abelian square, while if $\lfloor n \alpha \rfloor$ is odd, then every heavy factor of length $n$ contains an even number of $\sa{a}$'s and hence is an abelian square, whence the statement follows.
\end{proof}

Recall that given a finite or infinite word $w$,  $\ass{w} {n}$ denotes the number of distinct abelian-square factors of $w$ of length $n$. 

\begin{corollary}\label{cor:formula}
 Let $s_{\alpha}$ be a Sturmian word of angle $\alpha$.
For every positive even $n$, let $I_{n}=\{\{-i \alpha\}\mid 1\le i \le n\}$. Then
\[
\ass{s_{\alpha}} {n} =
\begin{cases}
 \# \{x\in I_{n} \mid x \leq \{-n \alpha\}\}, & \text{ if $\lfloor n \alpha \rfloor$ is even; }\\
 \# \{x\in I_{n} \mid x \geq \{-n \alpha\}\}, & \text{ if $\lfloor n \alpha \rfloor$ is odd. }
\end{cases}
\]
\end{corollary}

\begin{example}
 The factors of length $6$ of the Fibonacci word $F$ are, lexicographically ordered:
$\sa{aabaab, aababa, abaaba, ababaa, baabaa}$ (heavy factors), $\sa{baabab, babaab}$ (light factors).
The light factors,  whose number of $\sa{a}$'s is $\lfloor 6  \alpha \rfloor  =3$,  are not abelian squares; the heavy factors,  whose number of $\sa{a}$'s is $\lceil 6 \alpha \rceil=4$,  are all abelian squares.

We have $I_{6}=\{0.382, 0.764, 0.146, 0.528, 0.910, 0.292\}$ (values are approximated) and $6\alpha\simeq 3.708$, so $\lfloor 6\alpha \rfloor$ is odd. Thus, there are $5$ elements in $I_{6}$ that are $\geq \{-6 \alpha\}$, so by Corollary \ref{cor:formula} there are $5$ distinct abelian-square factors of length $6$.

The factors of length $8$ of the Fibonacci word are, lexicographically ordered:
$\sa{aabaabab}$, $\sa{aababaab}$, $\sa{abaabaab}$, $\sa{abaababa}$, $\sa{ababaaba}$, 
$\sa{baabaaba}$, $\sa{baababaa}$, $\sa{babaabaa}$ (heavy factors), $\sa{babaabab}$ (light factor).
The light factor,  whose number of $\sa{a}$'s is $\lfloor 8  \alpha \rfloor  =4$,  is an abelian square; the heavy factors,  whose number of $\sa{a}$'s is $\lceil 8 \alpha \rceil=5$,  are not abelian squares.
We have $I_{8}=\{0.382, 0.764, 0.146, 0.528, 0.910, 0.292, 0.674, 0.056\}$ (values are approximated) and $8\alpha\simeq 4.944$, so $\lfloor 8\alpha \rfloor$ is even. Thus, there is only one element in $I_{8}$ that is $\leq \{8 \alpha\}$, so by Corollary \ref{cor:formula} there is only one abelian-square factor of length $8$.

In Table \ref{tab:i} we report the first few values of the sequence $\ass{F} {n}$
for the Fibonacci word $F$.  More detailed analysis of $\ass{F} {n}$ can
be found, using the decision method we used for Thue-Morse,
in \cite{Du&Mousavi&Schaeffer&Shallit:2014,Du&Mousavi&Schaeffer&Shallit:2016}.
\end{example}

\begin{table}[bt]
\centering  
\begin{small}
\begin{raggedright}
\begin{tabular}{c|c *{30}{@{\hspace{2.5mm}}c}}
$n$\hspace{2mm} &0  & 2  & 4  & 6  & 8 & 10  & 12  & 14  & 16  & 18  & 20 & 22 & 24 & 26 & 28 & 30 & 32 & 34 & 36
\\
\hline \\
$\ass{F} {n}$\hspace{2mm} &0 & 1& 3& 5& 1& 9& 5& 5& 15& 3& 13& 13& 5& 25& 9& 15 & 25 & 1 & 27 \\
\hline \rule[0pt]{0pt}{12pt} 
\end{tabular}
\end{raggedright}\caption{\label{tab:i} The first few values of the sequence $\ass{F} {n}$ of the number of distinct abelian-square factors of length $n$ in the Fibonacci word $F=s_{\phi-1,\phi-1}$. See OEIS sequence A241674.}
\end{small}
\end{table}

Recall that every irrational number $\alpha$ can be uniquely written as a (simple) continued fraction as follows:
\begin{equation}\label{cf}
\alpha = 
a_0+{1\over\displaystyle a_1+
        {\strut 1\over\displaystyle a_2+
	          {\strut 1\over\displaystyle a_3+ 
		  \raisebox{-1ex}{$\ddots$} }}}
\end{equation}
where $a_{0}=\lfloor \alpha \rfloor$, and the infinite sequence $(a_{i})_{i\geq 0}$ is called the sequence of {\it partial quotients\/} of $\alpha$. The continued fraction expansion of $\alpha$ is usually denoted by its sequence of partial quotients as follows: $\alpha=[a_{0};a_{1},a_{2},\ldots ]$, and each of its finite truncations $[a_{0};a_{1},a_{2},\ldots,a_{k}]$ is a rational number $n_{k}/m_{k}$ called the $k$th convergent to $\alpha$. We say that an irrational $\alpha=[a_{0};a_{1},a_{2},\ldots ]$ has bounded partial quotients if and only if the sequence $(a_{i})_{i\geq 0}$ is bounded.  

The continued fraction expansion of $\alpha$ is deeply related to the exponent of the factors of the Sturmian word $s_{\alpha}$. 
The second author \cite{Mi89} proved that a Sturmian word of angle $\alpha$ has bounded exponent if and only if $\alpha$ has bounded partial quotients.

Since the golden ratio $\phi$ is defined by the equation $\phi=1+1/\phi$, we have from Equation \ref{cf} that $\phi=[1;1,1,1,1,\ldots]$ and therefore $\phi-1=[0;1,1,1,1,\ldots]$, so the Fibonacci word is an example of a Sturmian word with bounded exponent (actually, one can prove that the superior limit of the exponent of a factor of the Fibonacci word is $(2+\phi)$ \cite{MignosiPirillo}).

Now we prove that if $\alpha$ has bounded partial quotients, then the Sturmian word  $s_{\alpha}$ is (uniformly) abelian-square-rich. For this, we will use a result on the {\it discrepancy\/} of uniformly distributed modulo $1$ sequences from \cite{KuNi}. To the best of our knowledge, this is the first application of this result to the theory of Sturmian words, and we think that this correspondence might be useful for deriving other results on Sturmian words.

Let $\omega=(x_n)_{n\geq 0}$ be a sequence of real numbers. For a positive integer $N$ and a subset $E$ of the torus $I$, we define $A(E;N;\omega)$ as the number of terms $x_n$, $0\leq n\leq N$, for which $\{x_n\}\in E$. If there is no risk of confusion, we will write $A(E;N)$ instead of $A(E;N;\omega)$.

\begin{definition}
 The sequence $\omega=(x_n)_{n\geq 0}$ of real numbers is said to be 
 {\it uniformly distributed modulo $1$} if and only if for every pair $\gamma,\delta$ of real numbers with $0\leq \gamma<\delta\leq 1$ we have
 \[
 \lim_{N\to \infty} \frac{A([\gamma,\delta);N;\omega)}{N}=\delta-\gamma.
 \]
 \end{definition}
 
 \begin{definition}
 Let $x_0,x_1,\ldots, x_N$ be a finite sequence of real numbers. The number
 \[
 D_N=D_N(x_0,x_1,\ldots, x_N)=\sup_{0\leq \gamma<\delta\leq 1}\left| \frac{A([\gamma,\delta);N)}{N}-(\delta-\gamma)\right|
 \]
 is called the {\it discrepancy} of the given sequence. For an infinite sequence $\omega$ of real numbers the discrepancy $D_{N}(\omega)$ is the discrepancy of the initial segment formed by the first $N+1$ terms of $\omega$.
\end{definition}

The two previous definitions are related by the following result.

\begin{theorem}[\protect{\kern-0.5em}\cite{KuNi}]
 The sequence $\omega$ is uniformly distributed modulo $1$ if and only if $\lim_{N\to \infty}D_N(\omega)=0$.
\end{theorem}

An important class of  uniformly distributed modulo $1$ sequences is given by the sequence $(n\alpha)_{n\geq 0}$ with $\alpha$ an irrational number. The discrepancy of the sequence $(n\alpha)$ will depend on the finer arithmetical properties of $\alpha$. In particular, we have the following theorem, stating that if $\alpha$ has bounded partial quotients, then its discrepancy has the least order of magnitude possible.

\begin{theorem}[\protect{\kern-0.3em}\cite{KuNi}]
\label{theor:KN2}
 Suppose the irrational $\alpha=[a_0;a_1,\ldots]$ has  partial quotients bounded by $K$. Then the discrepancy $D_N(\omega)$ of $\omega=(n\alpha)$ satisfies $ND_N(\omega)=O(\log N)$. More exactly, we have
\begin{equation}
 ND_N(\omega)\leq 3+\left(\frac{1}{\log \phi}+\frac{K}{\log(K+1)}\right)\log N
\end{equation}
where $\phi$ is the golden ratio.
\end{theorem}

As a consequence, we obtain the following corollary:

\begin{corollary}\label{cor:disc}
  Suppose the irrational $\alpha=[a_0;a_1,\ldots]$ has bounded partial quotients. Then there exists a positive constant $C$ such that for every pair $\gamma,\delta$ of real numbers with $0\leq \gamma<\delta\leq 1$ and for every $n$, we have
 \[
A([\gamma,\delta);n;(n\alpha)) \geq  n(\delta-\gamma) - C\log n.
 \]

\end{corollary}

Now we use the results above to prove that Sturmian words with bounded exponent are abelian-square-rich.

\begin{theorem}\label{theor:lin}
 Let $s_{\alpha}$ be a Sturmian word of angle $\alpha$ such that $\alpha$ has bounded partial quotients. Then there exists a positive constant $C$ such that for every $n$ one has $\sum_{m\leq n} \ass{s_{\alpha}} {m}\geq Cn^2$.
\end{theorem}

\begin{proof}
For every even $n$, let $I'_{n}=\{\{i \alpha\}\mid 1\le i \le n\}$. By Corollary \ref{cor:formula} and basic arithmetical properties of the fractional part, we have
\[
\ass{s_{\alpha}} {n} = 
\begin{cases}
 \# \{x\in I'_{n} \mid x \geq \{n \alpha\}\}, & \text{ if $\lfloor n \alpha \rfloor$ is even; }\\
 \# \{x\in I'_{n} \mid x \leq \{n \alpha\}\}, & \text{ if $\lfloor n \alpha \rfloor$ is odd. }
\end{cases}
\]
So,
\begin{eqnarray}
&& \sum_{m\leq n}\ass{s_{\alpha}} {m}   \label{a}\\
&& \geq  \sum_{m\leq n,\ \lfloor m\alpha\rfloor \text{ odd }} \# \{ \{i \alpha\}\mid\{i \alpha\}\leq \{m \alpha\}, 1\leq i\leq m\} \label{b}\\
&& \geq \sum_{m\leq n,\ \lfloor m\alpha\rfloor \text{ odd},\ \{m \alpha\}\geq 1/2} \# \{ \{i \alpha\}\mid\{i \alpha\}\leq 1/2, 1\leq i\leq m\} \label{c}\\
&& = \sum_{m\leq n,\ \{m \alpha/2\}\geq 1/2,\ \{m \alpha\}\geq 1/2} \# \{ \{i \alpha\}\mid\{i \alpha\}\leq 1/2, 1\leq i\leq m\} \label{d}\\
&& \geq \sum_{m\leq n,\ \{m \alpha/2\}\geq 3/4} \# \{ \{i \alpha\}\mid\{i \alpha\}\leq 1/2, 1\leq i\leq m\} \label{e}\\
&& \geq \sum_{m\leq n,\ \{m \alpha/2\}\geq 3/4} \# \{ \{i \alpha/2\}\mid\{i \alpha/2\}\leq 1/4, 1\leq i\leq m\}   \label{f} \\
 && \geq \hspace{-3mm} \sum_{n/2\leq m\leq n,\ \{m \alpha/2\}\geq 3/4} \# \{ \{i \alpha/2\}\mid\{i \alpha/2\}\leq 1/4, 1\leq i\leq n/2 \} \label{g} \\
\label{1}
&& = \# \{ \{i \alpha/2\}\mid\{i \alpha/2\}<1/4, 1\leq i\leq n/2\}\cdot \hspace{-4mm} \sum_{n/2\leq m\leq n}  \{ m \mid \{m \alpha/2\}\geq 3/4 \}, \label{h}
\end{eqnarray}
where: 
\eqref{d} follows from \eqref{c} because for every integer $m$ one has that $\lfloor m\alpha\rfloor$ is odd if and only if $\{m \alpha/2\}\geq 1/2$; 
\eqref{e} follows from \eqref{d}  because for every integer $m$,
$\{m \alpha/2\}\geq 3/4$ implies $\{m \alpha\}\geq 1/2$;  
\eqref{f} follows from \eqref{e}  because for every integer $i$,
$\{i \alpha/2\}\leq 1/4$ implies $\{i \alpha\}\leq 1/2$;  
finally \eqref{h} follows from \eqref{g} because the cardinality of the first set is independent from the sum. The other (in-)equalities are obvious.

Since $\alpha$ has bounded partial quotients, so does $\alpha/2$ (see, e.g., \cite{LS97}), and we can apply Corollary \ref{cor:disc} to evaluate the two factors of \eqref{h}. Therefore, we have
\begin{eqnarray*}
&& \# \{ \{i \alpha/2\}\mid\{i \alpha/2\}<1/4, 1\leq i\leq n/2\} \\ 
&& = A([0,1/4);n/2;(n\alpha/2)) \\\
&& \geq (1/4)(n/2)-C_1\log n \\
&& = n/8-C_1\log n,
\end{eqnarray*}
for a positive constant $C_1$, and
\begin{eqnarray*}
&& \hspace{-3mm} \sum_{n/2\leq m\leq n} \{ m \mid \{m \alpha/2\}\geq 3/4 \} \\
&& =A([3/4,1);n;(n\alpha/2))-A([3/4,1);n/2;(n\alpha/2)) \\
&&  \geq (1/4)n-C_2\log n - (1/4)(n/2) +C_3\log n \\
&& = n/8 - C_4 \log n,
\end{eqnarray*}
for  constants $C_2,C_3,C_4$. From this, we can conclude that the product of the two factors of \eqref{h} is greater than a constant times $n^2$, as required.
\end{proof}

Next, we recall a result of J.~Cassaigne:

\begin{lemma}\cite[Prop.~5]{Ca99}
 The recurrence quotient $r_{\alpha}$ of a Sturmian word of angle $\alpha=[0;a_1,a_2,\ldots]$ is finite if and only if $\alpha$ has bounded partial quotients and, in this latter case, one has $2+\limsup a_i < r_{\alpha} <3+\limsup a_i$.
\end{lemma}

Thus, we have the following:

\begin{corollary}
 Let $s_{\alpha}$ be a Sturmian word of angle $\alpha$. Then  $s_{\alpha}$ has bounded exponent if and only if $s_{\alpha}$ is uniformly abelian-square-rich.
\end{corollary}

\begin{proof}
We know that $s_{\alpha}$  has bounded exponent if and only if $\alpha$ has bounded partial quotients, if and only if $s_{\alpha}$ is linearly recurrent. The statement then follows from Theorem \ref{theor:lin}, Lemmas \ref{lem:lin} and \ref{lem:un} and Proposition \ref{prop:new}.
\end{proof}

\section{Conclusions}

We proved that the Thue-Morse word is uniformly abelian-square-rich. We suspect that the technique we used for the proof can be generalized to some extent, and could be used, for example,  to prove that a subclass of fixed points of uniform substitutions are uniformly abelian-square-rich.

We also proved that Sturmian words of bounded exponent are uniformly abelian-square-rich (and the converse also holds). The proof we gave is based on a classical result on the discrepancy of the uniformly distributed modulo $1$ sequence $(n\alpha)_{n\geq 0}$, where $\alpha$ is the angle of the Sturmian word. To the best of our knowledge, this is the first application of this result to the theory of Sturmian words, and we believe that the correspondence we have shown might be useful for deriving other results on Sturmian words.

We leave open the question of determining whether $s_{\alpha}$ is not abelian-square-rich in the case when $\alpha$ has unbounded partial quotients. Notice that, by Proposition \ref{prop:new}, such $s_{\alpha}$ cannot be uniformly abelian-square-rich, since a Sturmian word has bounded partial quotients if and only if it has bounded exponent.

We mostly investigated binary words in this paper. We conjecture that binary words have the largest number of distinct abelian-square factors. More precisely, we propose the following conjecture.

\begin{conjecture}
 If a word of length $n$ contains $k$ distinct abelian-square factors, then there exists a binary word of length $n$ containing at least $k$ distinct abelian-square factors.
\end{conjecture}

A slightly different point of view from the one we considered in this paper consists in identifying two abelian squares if they have the same Parikh vector. Two abelian squares are therefore called \emph{inequivalent} if they have different Parikh vectors \cite{FSP97}. Sturmian words only have a linear number of inequivalent abelian squares. Nevertheless, a word of length $n$ can contain $\Omega(n\sqrt{n})$ inequivalent abelian squares \cite{warsaw2}. Computations support the following conjecture:

\begin{conjecture}[see \cite{Ry14,warsaw2}]
 Every word of length $n$ contains at most $\Theta(n\sqrt{n})$ inequivalent abelian squares.
\end{conjecture}

\section{Acknowledgments}

The authors warmly thank Julien Cassaigne for useful discussions about the
number of abelian-square factors in the Thue-Morse word, and the anonymous referees for their valuable comments, which allowed us to improve the presentation of the paper.

\newcommand{\noopsort}[1]{} \newcommand{\singleletter}[1]{#1}

\end{document}